%% file: main.tex
\documentclass[11pt,english]{article}
\usepackage{tikz}
\usepackage{amsmath,amssymb,amsthm}
\usepackage{multicol}

\usepackage{graphicx,color}
\usepackage{fullpage}
\usepackage[noblocks]{authblk}
\usepackage{tcolorbox}
\usepackage{babel}
\usepackage{wrapfig}
\usepackage{MnSymbol}
\usepackage{mdframed}
\usepackage{mathtools}
\usepackage{floatrow}
\usepackage{multirow}
\usepackage{palatino}

\usepackage[ruled,linesnumbered,vlined]{algorithm2e}
\usepackage{tablefootnote}
\usepackage{thm-restate}
\usepackage{bbding}
\usepackage{pifont}
\usepackage{nicefrac}
\usepackage{enumerate}
\DeclareMathOperator{\polylog}{polylog}
\DeclareMathOperator{\poly}{poly}
\DeclareMathOperator{\argmin}{argmin}
\newcounter{magicrownumbers}

\newcommand{\abs}[1]{\lvert #1 \rvert}

\definecolor{Darkblue}{rgb}{0,0,0.4}
\definecolor{Brown}{cmyk}{0,0.61,1.,0.60}
\definecolor{Purple}{cmyk}{0.45,0.86,0,0}
\definecolor{Darkgreen}{rgb}{0.133,0.543,0.133}

\usepackage[colorlinks,linkcolor=Darkblue,filecolor=blue,citecolor=blue,urlcolor=Darkblue,pagebackref]{hyperref}
\usepackage[nameinlink]{cleveref}
\usepackage[colorinlistoftodos,prependcaption,textsize=tiny]{todonotes}

\newtheorem*{theorem*}{Theorem}
\newtheorem{theorem}{Theorem}
\newtheorem{lemma}{Lemma}
\newtheorem{definition}{Definition}
\newtheorem{claim}{Claim}

\newtheorem*{question*}{Question}

\newtheorem*{conjecture*}{Conjecture}
\newtheorem{question}{Question}
\newtheorem{construction}{Construction}

\newcommand{\old}[1]{{}}

\hypersetup{
    colorlinks=true,
    linkcolor=blue,
    filecolor=violet,  
    citecolor=violet,
    urlcolor=cyan,
    pdftitle={},
    pdfpagemode=FullScreen,
}

\title{Geometric Bipartite Matching is in NC} 

\date{}

\author{Sujoy Bhore\thanks{Department of Computer Science \& Engineering, Indian Institute of Technology Bombay, India.\\ Email: sujoy@cse.iitb.ac.in}
 \quad
Sarfaraz Equbal\thanks{Department of Computer Science \& Engineering, Indian Institute of Technology Bombay, India.\\ Email: sequbal@cse.iitb.ac.in}
\quad
Rohit Gurjar\thanks{Department of Computer Science \& Engineering, Indian Institute of Technology Bombay, India.\\ Email: rgurjar@cse.iitb.ac.in, supported by SERB MATRICS grant MTR/2022/001009}
}

\begin{document}
\maketitle

\begin{abstract}

In this work, we study the parallel complexity of the Euclidean minimum-weight perfect matching (\textsf{EWPM}) problem. 
Here our graph is the complete bipartite graph $G$ on two sets of points $A$ and $B$ in $\mathbb{R}^2$ and the weight of each edge is the Euclidean distance between the corresponding points. 
The weighted perfect matching problem on general bipartite graphs is known to be in \textsf{RNC} [Mulmuley, Vazirani, and Vazirani, 1987], and Quasi-\textsf{NC} [Fenner, Gurjar, and Thierauf, 2016]. 
Both of these results work only when the weights are of $O(\log n)$ bits. 
It is a long-standing open question to show the problem to be in \textsf{NC}.

First, we show that for \textsf{EWPM}, a linear number of bits of approximation 
is required to distinguish between the minimum-weight perfect matching and other perfect matchings. 
Next, we show that the \textsf{EWPM} problem that allows up to $\frac{1}{\poly(n)}$ additive error, 
is in \textsf{NC}.
\end{abstract}

\section{Introduction}
\input{intro}

\section{Lower Bound}\label{Sec:lower_bound}
\input{lowerbound}

\section{Geometric Bipartite Matching}
\label{sec: geometric matching}

In this section, we study the parallel complexity of 
$\delta$-\textsf{EWPM}, and show that the problem lies in the class \textsf{NC}, 
for $\delta = 1/\poly(n)$. 

First of all, we assume that no three vertex points are colinear.  
 There is a simple fix to break colinearity by way of small perturbations in coordinates. 
 Specifically, for the $i$th vertex at point $(x_i, y_i)$, let us 
 assign its new coordinates to be $(x_i + i/K, y_i + i^2/K)$,
 where $K$ is a large enough number.
  This specific perturbation guarantees that no three points are colinear. 
  To see this, consider $i$th, $j$th, and $k$th vertices after the perturbation. 
  They will be colinear if and only if the following matrix has zero determinant. 
  \[
\begin{pmatrix}
  1 & 1 & 1 \\
  x_i + i/K &  x_j + j/K &   x_k + k/K \\
    y_i + i^2/K &  y_j + j^2/K &   y_k + k^2/K 
  \end{pmatrix}
\]

Consider the coefficient of the term $1/K^2$ in the determinant, which is $(i-j)(j-k)(k-i) \neq 0$.
Other terms in the determinant will be an integer multiple of $1/K$ and hence, cannot cancel this term,
when $K$ is large enough ($\poly(n)$). 
  This perturbation can cause additive error in the weights of perfect matchings, 
  but the error will remain bounded by $O(n^3/K)$. 
  Thus, the minimum weight perfect matching with respect to perturbed coordinates
  will be an EWPM up to a $1/\poly(n)$ additive error.  
  To make the coordinate integral, we can multiply them by $K$. 
 Now, give a brief overview of our ideas. 



Our main idea is to design an isolating weight assignment for the given graph and 
then use the MVV algorithm. 
Let $G$ be a complete bipartite graph of two sets of points $A$ and $B$ in $\mathbb{R}^2$. 
The MVV theorem asserts that if a graph has an isolating weight assignment, then the task of finding the minimum weight perfect matching in $G$ can be accomplished in \textsf{NC}. 


To construct an isolating weight assignment, we adopt the weight scheme introduced by Tewari and Vinodchandran~\cite{tewari2012greens}, which was designed specifically for planar bipartite graphs.
However, note that our graph is the complete bipartite graph and hence, far from planar. 
Our first key observation is that the union of minimum weight perfect matchings 
(with respect to Euclidean distances or even approximate distances) 
forms a planar subgraph. 
Then one can hope to use the Tewari and Vinodchandran~\cite{tewari2012greens} weight scheme on this planar subgraph. 
However, we cannot really compute this planar subgraph (i.e., the union of minimum weight perfect matchings). 
What proves to be useful is the fact that the Tewari-Vinodchandran weight scheme
is black-box, i.e, it does not care what is the underlying planar graph, it only needs to know the points in the plane where vertices are situated.
Finally, we combine the approximate distance function with the Tewari-Vinodchandran weight function on a smaller scale, and apply it on the complete bipartite graph. 
We show that this combined weight function is indeed isolating.



Towards showing the planarity of the union of minimum weight perfect matchings, 
first we establish a simple fact that
for any convex quadrilateral,
 there is a significant difference between the sum of diagonals and the sum of any opposite sides.

\begin{lemma} 
Consider a convex quadrilateral formed by a quadruple in an integer grid of size $N\times N$. The sum of lengths of its diagonals is larger than the sum of any two opposite sides.
And the gap between the two sums is at least $\frac{1}{4N^4}$.
\label{lem: Gap bound}
\end{lemma}
\input{quadrilateralFigure.tex}

\begin{proof}
    Intuitively, the sum of the diagonals will be larger than the sum of any two opposite sides because of triangle
    inequality (the diagonals combined with any two opposite edges form two triangles).
     The significance of this gap arises from the fact that if the points are from a grid and are not collinear, then the angle between any side and the diagonal cannot be infinitely small. 
    Formally, Let the four corners of the quadrilateral be $A, B, C, D$ (in cyclic order). See Figure~\ref{fig:quad}. 
    Let $O$ be the intersection point of the diagonals $AC$ and $BD$ (diagonals always intersect in a convex quadrilateral). 
    By triangle inequality, we have $\abs{AO} + \abs{OB} \geq AB$ and $\abs{CO} + \abs{OD} \geq CD$. Adding the two we get,
    \[\abs{AD} + \abs{BC} \geq \abs{AB} + \abs{CD}.\]
    Now, we lower bound the gap. Let $OP$ be the perpendicular from $O$ to the line $AB$.
    And let $\theta$ be the angle $\angle OAP$.
    \begin{equation}
        \abs{AO} - \abs{AP} = \abs{AP}(\sec \theta -1 ) \geq \abs{AP} \tan^2 \theta /2. 
        \label{eq:AO-AP}
    \end{equation}
    Let us lower bound $\tan \theta$ for $\theta$ being the angle between any three non-colinear points on the $N\times N$ grid. 
    Without loss of generality let $\theta$ be the angle between the integer vectors $(i,j)$ and $(k,\ell)$. 
    The we can write
    \[\tan^2 \theta = 1- \cos^2 \theta = 1 - \frac{(ik+jl)^2}{(i^2+j^2)(k^2+\ell^2)} \geq \frac{(i\ell - jk)^2}{(i^2+j^2)(k^2+\ell^2)} \geq \frac{1}{4N^4}.\]
    The last inequality follows because $i,j,k,l \leq N$ and the difference $i\ell - jk$, being integral, is at least one. 
    From equation~(\ref{eq:AO-AP}), we get $\abs{AO} - \abs{AP} \geq \abs{AP}/(8N^4)$.
    Similarly, we can write $\abs{BO} - \abs{BP} \geq \abs{BP}/(8N^4)$.
    Adding these two we get
    \[ \abs{AO}+\abs{BO} - \abs{AB} \geq \abs{AB}/(8N^4).\]
    We can write a similarly inequality for $\abs{CO}+\abs{DO}-\abs{CD}$. Adding the two, we get
    \[ \abs{AC}+\abs{BD} - \abs{AB}-\abs{CD} \geq (\abs{AB} +\abs{CD})/(8N^4) \geq 2/(8N^4).\]
    The last inequality follows because distance between two integral points is at least 1.

\end{proof}

\old{\begin{lemma}
    With respect to approximate distance, the gap between the sum of diagonals and the opposite sides of a quadrilateral is notably significant If the vertices of the quadrilateral are from $N$ × $N$ grid.
\label{lem: Gap for approx dist}
\end{lemma}

\begin{proof}
    
\end{proof} 
}

\subsection{Union of Near-Minimum Weight Perfect Matchings}\label{sec:union}

In this subsection, we establish our main lemma that in a geometric bipartite graph $G$, 
the union of near-minimum weight perfect matchings forms a planar subgraph of $G$. 
This allows us to use the Tewari and 
Vinodchandran~\cite{tewari2012greens} isolating weight scheme for planar bipartite graphs. 
We first define a near-minimum weight perfect matching. 
\begin{definition}
Let the vertices of the geometric bipartite graph lie in the $N \times N$ integer grid. 
A perfect matching is said to be of near-minimum weight if its weight is less
than $w^* + 1/(8N^4)$, where $w^*$ is the minimum weight of a perfect matching.
\end{definition}

  
\begin{lemma}  
For a geometric bipartite graph $G$ with vertices in the $N \times N$ integer grid,
the union of near-minimum weight perfect matchings forms a planar graph. 
\label{lem: union lemma}
\end{lemma}
\begin{proof}
%
Let $A \cup B$ be the bipartition of the vertices.
We will first show that no two edges in a near-minimum weight perfect matching $M$ cross each other.
For the sake of contradiction, let there be two edges $\{a_1,b_1\}$ and $\{a_2,b_2\}$ in $M$ that cross each other,
where $a_1, a_2 \in A$ and $b_1, b_2 \in B$. 
Since $G$ is a complete bipartite graph, every vertex of set $A$ must have an edge to every vertex of set $B$ in $G$. 
We can construct another matching $M'$ from $M$ by replacing the crossing edges $\{a_1,b_1\}$ and $\{a_2,b_2\}$ with $\{a_1,b_2\}$ and $\{a_2,b_1\}$, respectively.

Note that $(a_1, a_2, b_1, b_2)$ form a convex quadrilateral, since its diagonals $a_1b_1$
and $a_2b_2$ cross each other.
From Lemma~\ref{lem: Gap bound}, 
we know that 
\[ \abs{a_1b_2}+\abs{a_2b_1} \leq \abs{a_1b_1} + \abs{a_2b_2}   - 1/(4N^4).\]
From here, we can conclude that $w(M') \leq w(M) - 1/(4N^4)$,
where $w(M')$ and $w(M)$ are the weights of $M'$ and $M$, respectively.
This contradicts the fact that $M$ is a near-minimum weight perfect matching.

Now, we will show that two edges belonging to two different near-minimum weight perfect
matchings cannot cross. 
%
Consider two such near-minimum weight perfect matchings $M_1$ and $M_2$, 
where the edges $\{a_1,b_1\} \in M_1$ and $\{a_2,b_2\} \in M_2$ cross each other.
%
Observe that the union of these two perfect matchings forms a set of vertex-disjoint cycles, and a set of disjoint edges (which are common to both).
There are two cases: (i) the edges $\{a_1,b_1\}$ and $\{a_2,b_2\}$ are part of 
one of these cycles and (ii) they are part of two different cycles. 
In each of the cases, we will create two new perfect matchings 
with significantly smaller weight,
which will contradict the near-minimumness of $M_1$ and $M_2$. 

Case (i): $\{a_1,b_1\}$ and $\{a_2,b_2\}$ are part of one cycle $C$. See Figure~\ref{fig:crossing}.
Note that the edges of this cycle come alternatingly from $M_1$ and $M_2$ 
(shown in the figure in red and blue colors).

\input{uncrossingOneCycle.tex}
We construct two distinct perfect matchings, $M_1'$ and $M_2'$, using $M_1$ and $M_2$. 
Removing the edges $\{a_1,b_1\}$ and $\{a_2,b_2\}$ from cycle $C$ divides it into two parts.
Note that both parts must have even number of edges, since the edges are alternating
between $M_1$ and $M_2$. 
It follows that one of these parts is a path from $a_1$ to $a_2$, let us call it $C_1$.
And the other one 
is a path from $b_1$ to $b_2$, let us call it $C_2$ (as shown in Figure~\ref{fig:crossing}).

Let us put $\{a_1,b_2\}$ into $M_1'$ and $\{a_2,b_1\}$ into $M_2'$.
For the edges in $C_1$, we put  the $M_1$ edges into $M_1'$ and 
 the $M_2$ edges into $M_2'$. 
For the edges in $C_2$ we do the opposite, put the   $M_1$ edges into $M_2'$ 
and the $M_2$ edges into $M_1'$. 
For edges outside of the cycle $C$,
we put edges from $M_1$ into $M_1'$ and edges from $M_2$ into $M_2'$.

Case (ii): $\{a_1,b_1\}$ and $\{a_2,b_2\}$ are part of two different cycles. 
Let $C_1$ and $C_2$ be the paths obtained from removing $\{a_1,b_1\}$ and $\{a_2,b_2\}$
from the two cycles, respectively. See Figure~\ref{fig:crossing2}.
Here again we construct two distinct perfect matchings, $M_1'$ and $M_2'$, using 
a similar uncrossing of edges.
Let us put both $\{a_1,b_2\}$ and  $\{a_2,b_1\}$ into $M_1'$.
For the edges in $C_1$, we put  the $M_1$ edges into $M_1'$ and 
 the $M_2$ edges into $M_2'$. 
For the edges in $C_2$ we do the opposite, put the   $M_1$ edges into $M_2'$ 
and the $M_2$ edges into $M_1'$. 
For edges outside the two cycles,
we put edges from $M_1$ into $M_1'$ and edges from $M_2$ into $M_2'$.

Note that in both Case (i) and Case (ii),
the newly constructed perfect matchings $M_1'$ and $M_2'$ together have 
the same edges as $M_1 \cup M_2$, except for 
$\{a_1,b_1\}$ and $\{a_2,b_2\}$ being replaced
with $\{a_2,b_1\}$ and $\{a_1,b_2\}$.

Let $w_1, w_2, w'_1, w'_2$ be the weights of matchings $M_1, M_2, M'_1, M_2'$, repsectively. 
Then,
\[w_1' + w_2' = w_1+w_2 - \abs{a_1b_1} - \abs{a_2b_2}+  \abs{a_1b_2}+
\abs{a_2b_1}  .\]
From Lemma~\ref{lem: Gap bound}, we have that 
 \[ \abs{a_1b_1} + \abs{a_2b_2} -  \abs{a_1b_2}
-\abs{a_2b_1} \geq  1/(4N^4). \]
Thus, 
\[w_1' + w_2' \leq w_1+w_2 - 1/(4N^4) .\]
Let $w^*$ be the weight of the minimum weight perfect matching. 
Since $M_1$ and $M_2$ are of near-minimum weight, we have $w_1, w_2 < w^*+1/(8N^4) $.
Using this with the above inequality, we get $w'_1 + w'_2 < 2w^*$. 
This implies that at least one of the two matchings $M_1'$ and $M_2'$
have weight smaller than $w^*$, which is a contradiction.

\end{proof}


\subsection{Weight scheme}
Now, we come to the design of an isolating weight assignment for the graph and the proof of our main theorem.
One of the components of our weight scheme is the isolating weight assignment $W_{TV}$
constructed by Tewari and Vinodchandran~\cite{tewari2012greens} for planar bipartite graph. 
We will use the same weight scheme, but for any graph (not necessarily planar) embedded in the plane. 

Consider a bipartite graph $G = (A, B, E)$ (not necessarily planar) with a straight-line embedding in $\mathbb{R}^2$. 
For any vertex $u$, let $(x_u,y_u)$ be the associated point in $\mathbb{R}^2$. 
For an edge $e = (u, v)$, where $u \in A$ and $v \in B$, 
we define the weight function $W_{TV}$ as follows: 
\[ W_{TV}(e) = (y_v - y_u) \times (x_v + x_u) \]
Then, the theorem below says that $W_{TV}$ is isolating for bipartite planar graphs. 
\begin{theorem}[\cite{tewari2012greens}]
    Let G be a planar bipartite graph.
    Then with respect to weight function $W_{TV}$ (defined using any planar embedding), 
    then the minimum weight perfect matching in $G$, if one exists, is unique. 
\label{theo:raghunath}
\end{theorem}
%
For a geometric bipartite graph, 
our main idea is to combine $W_{TV}$ with the approximate distance function (up to a certain number of bits of precision)
The purpose of combining  $W_{TV}$ is to break ties among minimum weight perfect matchings
according to the approximate distance function. 
%

Let  $G(A, B, E)$ be  a geometric bipartite graph on the
$N \times N$ integer grid. 
Let $d(\cdot)$ be the weight function on the edges defined using the Euclidean distance
and let it naturally extend to subsets of edges.
For any positive integer $\ell$,
let us define the approximate distance function $d_\ell \colon E \to \mathbb{Z}$ 
 as 
 \[d_\ell (e) = \lfloor   d(e)  \times 2^\ell   \rfloor .\] 
First let us show that the minimum weight perfect matchings with respect to 
approximate distance function remain near-minimum with respect to the exact distance function.

\begin{claim}
    For any positive integer $\ell$, 
    let $M$ and $M^*$ be  minimum weight perfect matchings with respect to functions
    $d_\ell$ and $d$, respectively. 
    Then, \[d(M) < d(M^*) + n/2^\ell .\]
    \label{cla:approx}
\end{claim}
\begin{proof}
Observe that for any edge $e$, $ 2^\ell d(e) -1 < d_\ell (e)\leq 2^\ell d(e)$.
Hence, for  perfect matching $M$,
\[2^\ell d(M) -n < d_\ell(M) \leq 2^\ell d(M) .\]
 Then, we can write
\[ 2^\ell d(M) < d_\ell (M) + n \leq  d_\ell (M^*) + n \leq 2^ \ell d(M^*) + n   .\]
This implies that $d(M) < d(M^*) + n/2^\ell$. 
\end{proof}

\paragraph{Weight scheme.}
For any integer $\ell$, now
 let us define the  combined weight function $W_{\ell}$ on the edges as follows:
 \[ W_{\ell} := (2n N^2+1) \times d_\ell + W_{TV} \]
%
Here, the scaling $d_\ell$ with a large number ensures that 
$W_{\ell}$ has the same ordering of perfect matchings as $d_\ell$, 
and the $W_{TV}$ function plays the role of tie breaking. 
Our next lemma says that when to take enough number of bits from the 
distance function and then combine it with $W_{TV}$ as above,
the resulting weight function is isolating. 

\begin{lemma}  
For any integer $\ell \geq  4 \log N + \log n +3 $, 
the minimum weight perfect matching in $G$ with respect
to the weight function $W_\ell$ is unique.
\label{lem: uniqness}
\end{lemma}
\begin{proof} 
First observe that for any two perfect matchings $M_1$ and $M_2$,
\[d_\ell(M_1) > d_\ell(M_2) \implies W_\ell(M_1) > W_\ell(M_2) .\]
This is because the maximum contribution of $W_{TV}$ 
to the weight of a matching can be at most $n \times 2N^2$. 
Thus, we can write 
\begin{align*}
W_\ell(M_1) - W_\ell(M_2) &= (2n N^2+1) (d_\ell(M_1) - d_\ell(M_2)) + 
W_{TV}(M_1) - W_{TV}(M_2) \\
 &\geq (2n N^2+1) \cdot 1 + 0 - 2nN^2. \\
 &\geq 1
\end{align*}

It follows that the set of minimum weight perfect matchings with respect to $W_\ell$
is a subset of that with respect to $d_\ell$. 
Now, we argue that  these sets form a planar subgraph.

\begin{claim}
The union of minimum weight perfect matchings 
with respect to $d_\ell$ forms a planar subgraph. 
\end{claim}
\begin{proof}
Let $M$ and $M^*$ be the minimum weight perfect matchings with respect to 
functions $d_\ell$ and $d$, respectively.
From Claim~\ref{cla:approx} we have that $d(M) < d(M^*) + n/2^\ell$. 
By substituting $\ell = 4 \log N +  \log n +3$,
we get that the gap is less than $1/(8N^4)$.
Hence, $M$ is a near-minimum weight perfect matching with respect to the 
$d(\cdot)$.
Then,
the claim  follows from Lemma~\ref{lem: union lemma}. 
\end{proof}
To finish the proof of the lemma,
let $H$ be the subgraph formed by the union of minimum weight
perfect matchings with respect to $d_\ell$.
Clearly, $d_\ell$ gives equal weights to all the perfect matchings
in $H$. 
Thus, the function $W_{\ell}$ is same as $W_{TV}$ on $H$ (up to an additive constant).
From Theorem~\ref{theo:raghunath}, we know that $W_{TV}$ ensures
a unique minimum weight perfect matching in the planar graph $H$. 
Hence, so does $W_\ell$.
\end{proof}

\paragraph{Proof of the main theorem (Theorem~\ref{theo: Main theo})}
Once we have shown how to construct an isolating weight assignment, 
we just need to use the algorithm of 
Mulmuley, Vazirani and Vazirani~\cite{mulmuley1987matching} 
to construct the minimum weight perfect matching. 


\begin{theorem}[\cite{mulmuley1987matching}]
\label{mvv-thm}
Given a graph $G=(V,E)$ with an isolating weight assignment on the edges that uses $O(\log n)$ bits, 
there is an \textsf{NC} algorithm to  find the minimum-weight perfect matching.
\label{theo: MVV}
\end{theorem}

Now, we are ready to prove the main theorem. 
Suppose we are given a bipartite set of $2n$ points in $N \times N$ integer grid.
Recall that the weight of an edge is defined to be the Euclidean distance between the endpoints. 
Our goal is to construct a perfect matching whose weight is at most $w^* + \delta$,
where $\delta$ is the given error parameter and $w^*$ is the minimum weight of a perfect matching. 
If we choose $\ell \geq \log (n/ \delta )$, then from Claim~\ref{cla:approx},
we know that a minimum weight perfect matching with respect to function 
$d_\ell(\cdot)$ will have the desired property.

We choose $\ell = \max\{ \log(n/\delta) , 4 \log N +  \log n +3 \}$. 
Then we use the weight scheme $W_\ell$ with the MVV algorithm (Theorem~\ref{theo: MVV}).
Recall that from Lemma~\ref{lem: uniqness}, we have the isolation property required in 
Theorem~\ref{theo: MVV}. 
Finally, let us analyse the number of bits used by weight function $W_\ell$. 
The maximum weight given to any edge by function $d(\cdot)$ is at most $\sqrt{2}N$
and by function $W_{TV}$, it is at most $2N^2$.
Thus, maximum weight given to any edge by function $W_\ell$
will be at most $2^\ell \times \sqrt{2}N \times (2n N^2+1) + 2N^2 $. 
The number of bits in weight of any edge come out  to be $O(\log N  n / \delta)$. 
Hence, we have an \textsf{NC} algorithm, whenever $N$ and $1/\delta$
are polynomial in $n$. 






\section{Conclusion}
In this work, we explored the parallel complexity of \textsf{EWPM} problem. We established a lower bound which shows that for \textsf{EWPM}, a linear number of bits is required to distinguish the minimum-weight perfect matching from others. 
Next, we showed that \textsf{EWPM} problem that allows up to $\frac{1}{\poly(n)}$ error, is in \textsf{NC}. 
The main question arises from our work is whether 
the non-bipartite version of \textsf{EWPM}  is also in \textsf{NC}. 
Another possible extension is to consider the bipartite version in 3 or higher dimensions. 


\bibliographystyle{alphaurl}
\bibliography{references}

\end{document}

%% file: intro.tex
\old{

Introduction flow. 

\begin{itemize}
    \item Literature on NC. Mention - general graphs it is quasi-NC. Also, general bipartite graphs it is quasi-NC. 
    \item Geometric Matching. Can we use geometric properties to show that it is in NC. 
    \item Literature on geometric bipartite matching. But, we can only solve approx. In the next point, we justify.. 
    \item Discuss the issue of precision: sum of square roots, etc... E.g., EXACT TSP not known to be in NP, or Matching/Spanning tree not even known to be in P...  
\end{itemize}
}

The perfect matching problem is one of the well-studied problems in Complexity theory, especially, in the context of derandomization and parallelization.
Given a graph $G=(V,E)$, the problem asks,  whether the graph contains a matching that matches every vertex of $G$. 
Due to Edmonds~\cite{edmonds1965paths}, the problem is known to be solvable in polynomial time. 
However, the parallel complexity of the problem has not been completely resolved till today.
In 1979, Lov\'{a}sz~\cite{lovasz1979determinants} showed that perfect matching can be solved by efficient randomized  parallel algorithms, i.e., the problem is in 
\textsf{RNC}. 
Hence, the main question, with respect to its parallel complexity, is whether this randomness is necessary, i.e., whether the problem is in \textsf{NC}\footnote{The class \textsf{NC} represents the
problems that have efficient parallel algorithms, i.e., they have uniform circuits of polynomial size and $\polylog$ depth}. 

The search version of the problem asks to explicitly construct a perfect matching in a graph if one exists. 
Note that in the parallel setting, there is no obvious reduction from search to decision. 
This version is also known to be in \textsf{RNC}~\cite{karp1985constructing, mulmuley1987matching}. 
The Mulmuley-Vazirani-Vazirani (MVV) algorithm~\cite{mulmuley1987matching}, in fact, also  works
for the weighted version of the problem, 
where there is a polynomially bounded weight assignment given on the edges of the graph.

The MVV algorithm~\cite{mulmuley1987matching} introduced the celebrated \emph{Isolation lemma}.
A weight assignment is called \emph{isolating} for a graph $G$, if the minimum weight perfect matching in $G$ is unique, if one exists. 
Mulmuley, Vazirani, and Vazirani~\cite{mulmuley1987matching} showed that given an isolating weight assignment with polynomially bounded integer weights for a graph $G$, a perfect matching in $G$ can be constructed
in \textsf{NC}.
The only place where they use randomization is to get an isolating weight assignment.
Their Isolation lemma states that a random weight assignment is isolating!

\emph{Derandomizing} the Isolation lemma means to construct such a weight assignment deterministically in \textsf{NC}.
A line of work derandomized the Isolation Lemma for special families of graphs, e.g., planar bipartite graphs~\cite{datta2010deterministically, tewari2012greens}, strongly chordal graphs~\cite{dahlhaus1998matching}, graphs with a small number of perfect matchings~\cite{grigoriev1987matching}.
In 2016, Fenner, Gurjar, and Thierauf~\cite{fenner2016bipartite} showed that the bipartite perfect matching problem is in quasi-\textsf{NC}, by an almost complete derandomization of the Isolation Lemma. 
Later, Svensson and Tarnawski~\cite{svensson2017matching} showed that  
the problem in general graphs is also in Quasi-\textsf{NC}. 
Subsequently, Anari and Vazirani~\cite{AnariV20} gave an \textsf{NC} algorithm 
for finding perfect matching in general planar graphs.
All of these algorithms work for the weighted version (poly-bounded) of the problem as well. 

What remains a challenging open question is to find an \textsf{NC} algorithm 
for any versions (decision/search/weighted) of the perfect matching problem, 
even for bipartite graphs. 
Inspired by the positive results on planar bipartite graphs, we investigate 
the weighted version of the perfect matching problem in the geometric setting (2 dimensional). 


\paragraph*{Geometric Bipartite Matching.} 
Let $A$ and $B$ be two point sets in $\mathbb{R}^2$ of size $n$ each. 
Consider the \textbf{complete} bipartite graph $G(A,B, E)$ with the following cost function on the edges: 
for any edge $e=(a,b)$, define $\mathcal{C}(e)= ||a-b||$,  where $||\cdot||$ denotes the Euclidean norm. 
In other words, we consider the Euclidean distance between the endpoints as the cost of an edge. 
The cost of a perfect matching $M$ is the sum of its edge costs $\mathcal{C}(M)=\Sigma_{e\in M} \mathcal{C}(e)$. 
The Euclidean minimum-weight perfect matching (\textsf{EWPM}) problem
is to find 
$M_{opt}=\argmin_{|M|=n}\mathcal{C}(M)$,
that is, the optimal perfect matching with respect to function $\mathcal{C}$.
\textsf{EWPM} 
is a fundamental problem in Computational Geometry and has been studied extensively over the years. 
See Section~\ref{related_work} for an overview of the results.
In this work, we focus on the parallel complexity of \textsf{EWPM} problem, and ask the following question - 

\begin{question}\label{ewpm-nc}
    Is \textsf{EWPM} in \textsf{NC}?
\end{question}

Optimization problems in computational geometry are usually studied in 
real arithmetic computational model, where comparing two distances
or sums of distances is assumed to be a unit cost operation.
However, in the bit complexity model, it is not clear if distances,
which can be irrational numbers, can be efficiently added or compared. 
In fact, the problem of comparing a two sums of square roots is not known to be in \textsf{P} (see, for example,~\cite{o1981advanced, TOPP}). 
See~\cite{EHS24} for some recent progress on the sum of square roots problem. 

Our path towards showing an \textsf{NC} algorithm for \textsf{EWPM}  
naturally goes via the MVV algorithm.
Recall that the MVV algorithm works only when the given weights/costs are polynomially bounded integers, because in intermediate steps, it needs to put weights in the exponent.
Hence, inevitably we need to consider the bit complexity of the weights. 
Note that there are other parallel algorithms for the weighted perfect matching problem (e.g.,~\cite{Goldberginterior92}), 
but there too it is important that the weights are polynomially bounded integers. 

It is not clear if the \textsf{EWPM} problem is in \textsf{P} (or even in \textsf{NP})
in the bit-complexity model. 
To the best of our knowledge, the existing algorithms for \textsf{EWPM} require the comparison between two 
sums of square roots. 
This naturally leads us to consider an approximate version of the problem. 
Let us define the $\delta$-\textsf{EWPM} problem, which asks for the Euclidean minimum-weight perfect matching up to an additive error $\delta$.  
We aim to get an \textsf{NC} algorithm for the problem whenever $1/\delta$ is $\poly(n)$. 





\subsection{Our Contribution}
In this work, we study the parallel complexity of $\delta$-\textsf{EWPM} problem. 
First, it is natural to ask whether solving $\delta$-\textsf{EWPM} for some $\delta =1/\poly(n)$
will already solve the \textsf{EWPM} problem.
In other words, by considering  $O(\log n)$ bit approximations of Euclidean distances, can we hope to find the Euclidean minimum weight perfect matching?
Our first result rules out this possibility. 
We show that for \textsf{EWPM}, a super-linear number of bit approximations is required to distinguish the minimum-weight perfect matching from others. 

\begin{theorem}
There is a set of $2n$ points in the $O(n^5) \times O(n^5)$ integer grid
such that in the corresponding complete bipartite graph,
the difference between the weights of 
the minimum weight perfect matching and another perfect matching
is at most $1/(n-1)!$. 
\label{thm:lower bound}
\end{theorem}
This theorem is proved in Section~\ref{Sec:lower_bound}. 

Next, we come to our positive result. 
We affirmatively answer Question~\ref{ewpm-nc}, by showing that the Euclidean minimum weight perfect matching problem that allows up to $\frac{1}{\poly(n)}$ error, 
is in \textsf{NC}.
\begin{theorem}
The $\delta$-\textsf{EWPM} is in \textsf{NC},
when $1/\delta$ is $\poly(n)$ and the points are on a polynomially bounded integer grid.
\label{theo: Main theo}
\end{theorem}
This theorem is proved in Section~\ref{sec: geometric matching}.



\subsection{Related Work}\label{related_work}
The classical Hopcroft-Karp algorithm computes
a maximum-cardinality matching in a bipartite graph with $n$ vertices and $m$ edges in $O(m\sqrt{n})$ time~\cite{hopcroft1973n}. 
After almost three decades, Madry~\cite{madry2013navigating} improved the running time to $O(m^{10/7}\polylog n)$ time, which was further improved to 
$O(m+n^{3/2}\polylog n)$ by Brand et al.~\cite{van2020bipartite}. 
The Hungarian algorithm computes the minimum-weight maximum cardinality matching in $O(mn+n^{2}\log n)$ time~\cite{papadimitriou1998combinatorial}. 
In a recent breakthrough, Brand et al.~\cite{van2023deterministic} showed that maximum-cardinality matching in bipartite graphs, can be solved in near-linear time.

For two sets of points $A$ and $B$ in $\mathbb{R}^2$, 
the best known algorithm for computing \textsf{EWPM} runs in 
$O(n^2\polylog n)$ time~\cite{agarwal2019efficient, agarwal1995vertical}. Moreover, if points
have integer coordinates bounded by $\Delta$, the running time can be improved to $O(n^{3/2}\polylog n \log \Delta)$~\cite{sharathkumar2013sub}.
If coordinates of input points have real values, it is not known whether a subquadratic algorithm exists for computing EMWM. However, for the non-bipartite case, Varadarajan~\cite{varadarajan1998divide} presented an $O(n^{3/2}\polylog n)$-time algorithm under any $\ell_p$-norm. For bipartite matching, a large body of literature focused on obtaining approximate matching for points in $\mathbb{R}^d$. Varadarajan and Agarwal~\cite{varadarajan1999approximation} presented an $O(n^{3/2}\varepsilon^{-d}\log^d n)$-time $\varepsilon$-approximation algorithm for
computing \textsf{EWPM} of points lying in $\mathbb{R}^d$.
Later, Agarwal and Raghavendra~\cite{sharathkumar2012algorithms} improved the running time. Recently, Agarwal et al.~\cite{agarwal2022deterministic} presented a deterministic algorithm with running time $n\cdot (\varepsilon^{-1}\log n)^{O(d)}$ time, and computes a perfect matching whose cost is within a $(1+\varepsilon)$ factor of the optimal matching under any $\ell_p$-norm.

\old{
We are interested in the geometric perfect matching problem for both bipartite and non-bipartite versions. 

Open Problems and ideas. 

\begin{enumerate}
    \item Ques. For points in $\mathbb{R}^d$, Geometric Minimum Weight Perfect Matching is in \textsf{NC}? - We don't know this answer even for $d=2$. 

    \item For points in $\mathbb{R}^d$, Geometric Minimum-Weight Bipartite Perfect Matching is in \textsf{NC}? We can show for points in the plane this is True. 

    \item Does there exist a lower bound? - We can show an existential lower bound that shows, in order to distinguish between minimum weight perfect matching and next one, it would require at least ... bits.

    \item Can we show an upper bound on the number of bits? - Currently, we are trying to show such an upper bound, which would also mean that the problem is in class \textsf{P}, which is surprisingly unknown till date. 
\end{enumerate}

}

\old{
The parallel complexity of \textsf{EWPM}, has not been studied in the bit-computation model. Moreover, in the Euclidean setting, under $\ell_2$-norm, we need to deal with square roots. Hence, the problem is closely related to the sum of squares root problem\footnote{In the sum of square root problem, the objective is to find 
 tight upper and lower bounds on the minimum positive values (say, $r(n,k)$) of $|\sum_{i=1}^{k} \sqrt{a_i} - \sum_{i=1}^{k} \sqrt{b_i}|$, where $a_i$ and $b_i$ are integers no larger than $n$.}, which is not known to be in $\mathcal{P}$ till date; see~\cite{o1981advanced} and ~https://topp.openproblem.net/p33
 \todo{S: how to cite this?}.
 The sum of square roots problem has been extensively studied over the years in Complexity theory; see~\cite{kayal2012sum, balaji2024ussr, jindal2023order, balaji2022identity}.

In the bit-computation model, the problem of distinguishing between two sums of square roots is not known to be in \textsf{P}. Of particular interest, it remains a key bottleneck to deciding for some classical optimization problems in Computational Geometry, e.g., Euclidean matching, shortest paths, etc, whether they are solvable in polynomial time in the bit-computation model.
 }

%% file: lowerbound.tex
In this section, we want to show that for a geometric bipartite graph with $n+n$ vertices, 
we need at least $\Omega(n \log n)$ bits of precision to distinguish the minimum weight perfect matching from others. 
We will show this by constructing a bipartite set of $2n$ points in the integer grid of size $O(n^5) \times O(n^5)$
 such that the difference between the weights of the minimum weight perfect matching and the one with the next higher weight will be $1/(n-1)!$.
Towards this, the first step is to construct a geometric graph where there are 
two perfect matchings whose weights differ by at most $ 1/(n-1)!$ (Claim~\ref{cla:TwoClosePM}).
Here we use an argument based on the pigeonhole principle. 
A similar argument  was used to show such a bound on the difference of two sums of square roots~\cite{ChengYuHsin10}. 

In the above construction, it is not necessary that one of the two perfect matchings is of minimum weight.  
In the second step, we show that the above geometric graph can be modified to construct
another one where the same two perfect matchings appear,
but now one of them is of minimum weight (Claim~\ref{cla:MinClosePM}).

\begin{construction}
Consider the left hand side vertices $u_0,u_1, \dots, u_{n-1}$ at points
\[\{ (0, 0), (0, 1), (0, 2), \dots, (0, n-1)  \} .\]
Similarly, consider the  right hand side vertices
$v_0,v_1, \dots, v_{n-1}$ at points
\[\{ (q, n), (q, 2n), (q, 3n), \dots, (q, n^2 )  \} , \]
where $q=n^4$.
\label{con:bipartitegraph}
\end{construction}

\begin{claim}
The geometric bipartite graph in Construction~\ref{con:bipartitegraph}
has all its perfect matchings with distinct weights. 
    \label{cla:distinct}
\end{claim}
\begin{proof}

Now, we will show that weights of any 
two distinct perfect matchings are different. 
Recall that edge weights are square roots of integers. 
We will argue that the edge weights
are linearly independent over rationals, which immediately
implies that any two different subsets of edges cannot have equal weights. 
It is known that to show linear independence of a set of square roots of integers,
it suffices to show that they are \emph{pairwise}
linearly independent 
(see, for example, \cite{CarOSul07}). 
So, now we just argue that the edge weights are pairwise linearly independent. 

For the sake of contradiction, 
suppose we have two edges $e$ and $e'$, whose weights
are linearly dependent.
Then we have $a w(e) = b w(e')$
for some integers $a$ and $b$. 
From here we get that $w(e)^2 w(e')^2 = (a/b)^2 w(e)^4$.
That is, the product $w(e)^2 w(e')^2$ is square of a rational number.
Since it is an integer, it must be square of an integer. 
From our construction, for any edge $e$, we have $q \leq w(e) \leq \sqrt{q^2 + n^4}$. 
Hence, $w(e)^2 w(e')^2$ must be square of an integer which between $q^2$ and $q^2 + n^4$. Moreover, we notice that only one edge, $e$ or $e'$, can be equal to $\sqrt{q^2 + n^4}$. So we can write $w(e)^2 w(e')^2 = (q^2 + \alpha)^2$ for some integer $ 0 \leq \alpha < n^4$.

For any edge $e$,  let us denote by $\Delta_e$,
the difference in the $y$ coordinates of the 
two endpoints of the edge. 
Then, the weight of an edge $e$ can be written as $w(e) = \sqrt{q^2 + \Delta_e^2}$. 
Now, we have 
\[ w(e)^2 w(e')^2 = (q^2 + \Delta_e^2 )(q^2 + \Delta_{e'}^2) = (q^2 + \alpha)^2. \]
Equivalently,
\[ q^2 ( \Delta_{e}^2  + \Delta_{e'}^2) + \Delta_{e}^2  \Delta_{e'}^2
= 2 q^2 \alpha + \alpha^2. \]
Observe that $\Delta_{e}^2  \Delta_{e'}^2 < n^8 = q^2$ (from construction as both $\Delta_{e}$ and  $\Delta_{e'}$ can not be equal to $n^2$)
 and also $\alpha^2 < n^8 = q^2$.
 Hence, we conclude from above that 
 \[ \Delta_{e}^2  + \Delta_{e'}^2 = 2 \alpha \text{ and }
 \Delta_{e}^2  \Delta_{e'}^2
=  \alpha^2. \]
This implies that  $\Delta_{e} = \Delta_{e'}$.

Now, we will argue that for any two distinct 
edges, we have $\Delta_{e} \neq  \Delta_{e'}$,
which will give us a contradiction. 
Indeed for the edge $(u_i,v_j)$, 
we have $\Delta_e =  jn -i$, which comes from a unique
choice of $0\leq i \leq n-1$  and $1\leq j \leq n$.

\end{proof}

\begin{claim}
In the geometric bipartite graph from Construction~\ref{con:bipartitegraph},
there are two perfect matchings whose weights are different and differ by at most
$1/(n-1)!$.
\label{cla:TwoClosePM}
\end{claim}
\begin{proof}
 From Claim~\ref{cla:distinct},
all $n!$ perfect matchings have distinct weights. 
From the construction, any perfect matching has its weight between $n^5 $ and $n\sqrt{n^8 + n^4} \leq n(n^4+1) $.
The bound follows from the pigeonhole principle.
\end{proof}

Now, consider the two perfect matchings from Claim~\ref{cla:TwoClosePM}, say $M_1 $ and $M_2$,
whose weights differ by at most $1/(n-1)!$ . 
Let $M_1$ be the one with a smaller weight.
The union of two perfect matchings $M_1 \cup M_2$ is a set of vertex-disjoint cycles and edges. 
We are going to ignore the common edges between $M_1$ and $M_2$.
Let $(e_1, e_2, \dots, e_{2\ell})$ be the sequence of edges generated from the cycles in $M_1 \cup M_2$ as follows: arrange the cycles in an arbitrary order. 
For each cycle, start from that edge in $M_1$ which has its left endpoint with minimum $y$-coordinate, and traverse along the cycle till we hit the starting vertex. 
Note that the sequence $(e_1, e_2, \dots, e_{2\ell})$ has edges alternating from $M_1$ and $M_2$. 
The new graph will be constructed by ``unrolling'' these cycles. 
The construction will be such that edges outside these cycles will be long, and hence, will not be a part of
  any minimum weight perfect matching. 
Recall that for any edge $e=(u_i,v_j)$ (Construction~\ref{con:bipartitegraph}),
we denote by $\Delta_e$ the difference in the $y$-coordinates of the endpoints, i.e., $ jn-i$. 
\begin{construction}
 Consider the vertex $t_0$ at $(0,  0)$. 
 Let $y_0 = 0$. 
 For $1 \leq k \leq 2 \ell$, we place
 the vertex $t_k$ at $( kq, y_k)$, where
 \begin{itemize}
 \item $y_k = y_{k - 1} + \Delta_{e_{k}}$ if $k$ is odd
 \item $y_k = y_{k - 1} - \Delta_{e_{k}}$ if $k$ is even
 \end{itemize}
 We add three more vertices: $s_{0}$ at $(0,-2\ell q)$,
 $s_{1}$ at $(\ell q, -2\ell q)$, and $s_{2}$ at $(2\ell q,-2\ell  q)$.
 See Figure~\ref{fig:unroll}.
 \label{con:unroll}
\end{construction}

   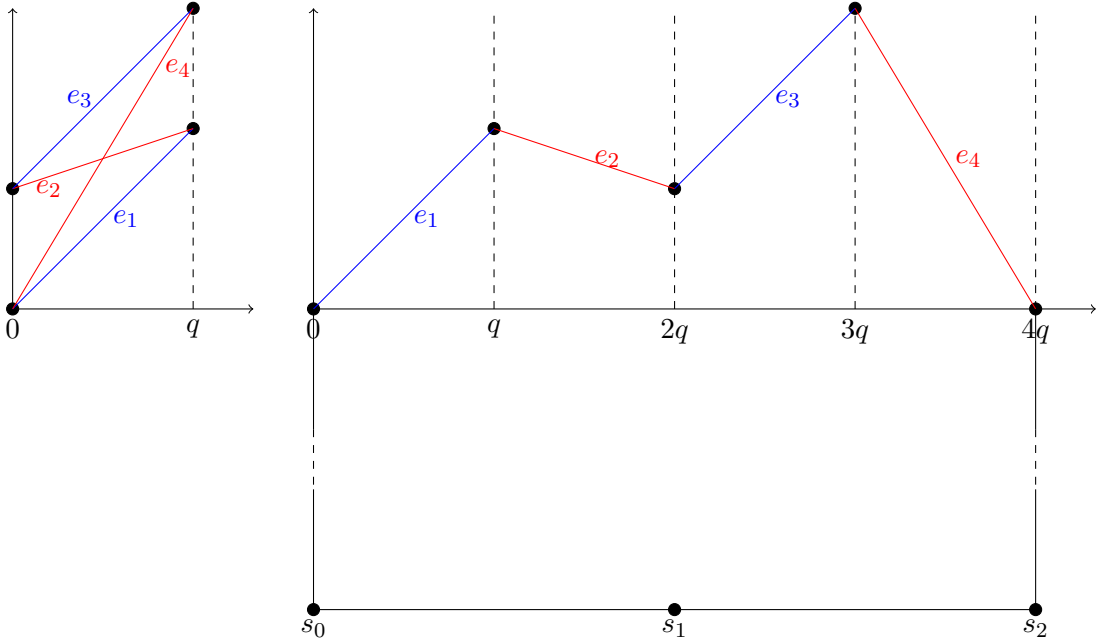
\begin{figure}[h!]
		\centering 
		\begin{tikzpicture}[scale=0.8]
		\draw[->] (0,0)--(0,5);
		\draw[dashed] (3,0)--(3,5);
		\draw[->] (0,0)--(4,0);		

		\node at (0,0) [below] {$0$}; 
		\node at (3,0) [below] {$q$}; 

		\draw [fill] (0,0) circle [radius=0.1];
		\draw [fill] (3,3) circle [radius=0.1];			
		\draw [fill] (0,2) circle [radius=0.1];
		\draw [fill] (3,5) circle [radius=0.1];
		
		\draw[blue] (0,0)-- node [right] {$e_1$} (3,3);
		\draw[red] (3,3)--(0,2);
		\draw[blue] (0,2)--node [left] {$e_3$} (3,5);
		\draw[red] (3,5)--(0,0);

		\node[red] at (0.6,2) {$e_2$};
		\node[red] at (2.75,4)  {$e_4$};
		
		\draw [fill] (5,0) circle [radius=0.1];
		\draw [fill] (8,3) circle [radius=0.1];			
		\draw [fill] (11,2) circle [radius=0.1];
		\draw [fill] (14,5) circle [radius=0.1];
		\draw [fill] (17,0) circle [radius=0.1];		
		
		\draw[->] (5,0)--(5,5);
		\draw[dashed] (8,0)--(8,5);
		\draw[dashed] (11,0)--(11,5);
		\draw[dashed] (14,0)--(14,5);
		\draw[dashed] (17,0)--(17,5);

		\draw[->] (5,0)--(18,0);		

		\node at (5,0) [below] {$0$}; 
		\node at (8,0) [below] {$q$}; 
		\node at (11,0) [below] {$2q$}; 
		\node at (14,0) [below] {$3q$}; 
		\node at (17,0) [below] {$4q$}; 
		
		\draw[blue] (5,0)-- node [right] {$e_1$} (8,3);
		\draw[red] (8,3)-- node [right] {$e_2$} (11,2);
		\draw[blue] (11,2)--  node [right] {$e_3$} (14,5);
		\draw[red] (14,5)-- node [right] {$e_4$} (17,0);

		\draw (17,0)--(17,-2); 
		\draw[dashed] (17,-2)--(17,-3); 
		\draw (17,-3)--(17,-5); 	
		\draw (17,-5)--(5,-5);
		\draw (5,0)--(5,-2); 
		\draw[dashed] (5,-2)--(5,-3); 
		\draw (5,-3)--(5,-5); 	

        \draw [fill] (17,-5) circle [radius=0.1];		
		\draw [fill] (5,-5) circle [radius=0.1];		
		\draw [fill] (11,-5) circle [radius=0.1];		
		
		\node at (5,-5) [below] {$s_0$};
		\node at (11,-5) [below] {$s_1$};
		\node at (17,-5) [below] {$s_2$};				

\end{tikzpicture}

	\caption{The left hand side figure shows a cycle in the union of two perfect matchings. The right hand side figure shows how we ``unroll" this cycle. }
\label{fig:unroll}
\end{figure}

Corresponding to perfect matchings $M_1$ and $M_2$,
here we will have prefect matchings $M_1'$ and $M_2'$ as
\[M_1' = \{ e_1, e_3, \dots, e_{2\ell-1},  (t_{2 \ell}, s_2), (s_0,s_1)\}\]
\[M_2'  = \{ e_2, e_4, \dots, e_{2\ell},  (t_{0}, s_0), (s_1,s_2)\} .\]
The following are easy observations about 
Construction~\ref{con:unroll}.

\begin{enumerate}[(I)]

   \item The edge lengths of $e_1, e_2, \dots, e_{2\ell}$  are exactly same as their lengths in Construction~\ref{con:bipartitegraph}. 
    
    \item $y_k \geq 0$ for each $1\leq k \leq 2\ell$, because for each cycle, the cycle traversal starts from the lowest $y$ coordinate on the left. Moreover, $y_{2\ell}$ must be zero, because any cycle traversal ends at the starting vertex. 

    \item Any pair of vertices are at least distance $q$ apart.

    \item $w(M_1') = w(M_1) + 3\ell q $ and $w(M_2') = w(M_2) + 3\ell q $. 
\end{enumerate}

\begin{claim}
The minimum weight perfect matching in Construction~\ref{con:unroll} is $M_1'$,  with weight   $w(M_1) + 3\ell q $. 
\label{cla:M1min}
\end{claim}
\begin{proof}
Recall that weight of any edge $e_k$ is at most $\sqrt{q^2 + n^4} = \sqrt{n^8 + n^4} \leq n^4 + 1/2 = q +1/2$. 
Hence, $w(M'_1) \leq \ell (q + 1/2) + 3 \ell q = \ell(4 q + 1/2)$.
We have already assumed that $M_2'$ has weight higher than $M'_1$.
Now, consider any perfect matching $M$ other than $M'_1$ and $M_2'$. We will consider different cases and argue that 
in each case $M$ has a larger weight. 

\begin{itemize}
\item  If $M$ matches  $ s_1$ with one of the $t_k$ vertices,
the weight of that edge will be at least $2 \ell q$. 
The vertices $s_0$ and $s_2$ will either match with each other or to some $t_k$ vertices. 
In either case, they will contribute at least $2 \ell q$ to the weight. 
The remaining vertices must have at least $\ell - 3$ edges, each with weight at least $q$. 
Hence, the total weight will be at least $5\ell q - 3q$,
 which is larger than  $w(M'_1)$.

\item Consider the case when $M$ has $(s_1, s_2)$ (the other case is similar) 
and  $s_0$ is matched with one of the $t_k$ vertices, other than $t_0$. 
 %
Recall that $y_k \geq 0$ and $s_0 = (0, -2\ell q)$.
 Then the weight of $(s_0, t_k)$ (for $k>0$)
 is at least $\sqrt{4\ell^2 q^2 + q^2} \geq 2\ell q + q/(4\ell) $. 
 The remaining $2 \ell$ vertices will have $\ell$ matching edges, each with weight at least $q$. 
Hence, the weight of the matching $M$ will be at least $\ell q + 2 \ell q + q/(4 \ell) + q\ell$. 
 This is clearly larger than $w(M'_1) \leq 4\ell q +\ell/2 $ (as $q=n^4$ and $\ell \leq n$).

\item Consider the case when $M$ has $(s_1, s_2)$ and $(s_0, t_0)$. 
These two edges will add up to weight $3 \ell q$. 
Since the matching $M$ is different from $M'_1$ and $M'_2$,
it must match a vertex $t_k$ with another vertex $t_j$ such that $j \neq \{k-1, k+1\}$. 
Then $\lvert j-k \rvert$  must be at least $3$, because the graph is bipartite.  The edge $(t_k, t_j)$ will have weight at least $3q$.
The other $\ell -1$ edges will have weight at least $q$. 
Hence, the total weight is at least $4 \ell q + 2q$, which is again larger than $w(M_1')$.

\item The other cases when $M$ has $(s_0, s_1)$ matched are similar to the above two cases. 

\end{itemize}

\end{proof}

Now, we finally come to our main claim. 

\begin{claim}
In the geometric bipartite graph from Construction~\ref{con:unroll},
the difference between the minimum weight perfect matching and the perfect matching with the next higher weight 
is at most $1/(n-1)!$.
\label{cla:MinClosePM}
\end{claim}
\begin{proof}
From Claim~\ref{cla:M1min}, we know that $M_1'$ is the minimum weight perfect matching. 
We had observed that $w(M_1') - w(M_2') = w(M_1)-w(M_2)$.
From Claim~\ref{cla:TwoClosePM}, this difference is at most $1/(n-1)!$.
\end{proof}

%% file: quadrilateralFigure.tex
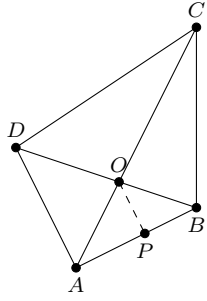
\begin{figure}
\centering
\begin{tikzpicture}[scale=0.8]

\coordinate (A) at (0,0);
\coordinate (B) at (2,1);
\coordinate (C) at (2,4);
\coordinate (D) at (-1,2);
\coordinate (O) at (5/7, 10/7);
\coordinate (P) at  (8/7, 4/7);  

\draw (A) -- (B) -- (C) -- (D) -- cycle;

\draw (A) -- (C);
\draw (B) -- (D);

\draw[dashed] (O) -- (P);

\foreach \point in {A,B, P}
    \filldraw (\point) circle (2pt) node[below] {\scriptsize $\point$};

\foreach \point in {C,D,O}
    \filldraw (\point) circle (2pt) node[above] {\scriptsize $\point$};

\end{tikzpicture}
\caption{A convex quadrilateral with its two diagonals.}
\label{fig:quad}
\end{figure}

%% file: uncrossingOneCycle.tex
\newcommand{\x}{3}
\newcommand{\y}{5}
\newcommand{\z}{8}
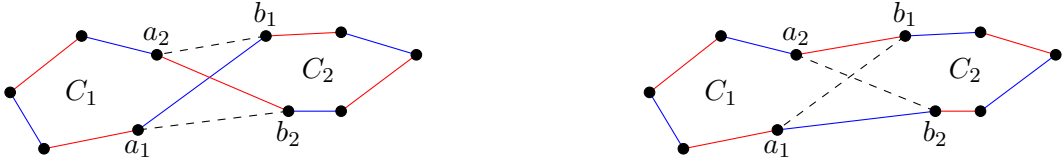
\begin{figure*}[h]
    \centering\begin{tikzpicture}
  \coordinate (a1) at (\x,\y);
  \coordinate (b2) at (\x+2,\y+0.25);
  \coordinate (b1) at (\x+1.7,\y+1.25);
  \coordinate (a2) at (\x+0.25, \y+1);

  \coordinate (a1_1) at (\x-1.25,\y-0.25);
  \coordinate (b1_1) at (\x-1.7,\y+0.5);
  \coordinate (b1_2) at (\x-0.75,\y+1.25);

  \coordinate (a2_1) at (\x+2.7,\y+1.3);
  \coordinate (a2_2) at (\x+3.7,\y+1.0);
  \coordinate (b2_1) at (\x+2.7,\y+0.25);

  \coordinate (C_1) at (\x-0.75,\y+0.5);
  \coordinate (C_2) at (\x+2.4,\y+0.80);


  \coordinate (a11) at (\x+8.5,\y);
  \coordinate (b22) at (\x+10.6,\y+0.25);
  \coordinate (b11) at (\x+10.2,\y+1.25);
  \coordinate (a22) at (\x+8.75, \y+1);

  \coordinate (a11_1) at (\x+7.25,\y-0.25);
  \coordinate (b11_1) at (\x+6.8,\y+0.5);
  \coordinate (b11_2) at (\x+7.75,\y+1.25);
  
  \coordinate (a22_1) at (\x+11.2,\y+1.3);
  \coordinate (a22_2) at (\x+12.2,\y+1.0);
  \coordinate (b22_1) at (\x+11.2,\y+0.25);

  \coordinate (C_11) at (\x+7.75,\y+0.5);
  \coordinate (C_22) at (\x+11,\y+0.80);

    \draw[blue] (a1) -- (b1);
    \draw[red] (a2) -- (b2);
    \draw[red] (a1) -- (a1_1);
    \draw[red] (b1_1) -- (b1_2);
    \draw[blue] (a2) -- (b1_2);
    
    \draw[red] (b1) -- (a2_1);
    \draw[blue] (a2_1) -- (a2_2);
    \draw[blue] (b2) -- (b2_1);

    \draw[blue] (b1_1) -- (a1_1);
    \draw[red] (b2_1) -- (a2_2);
    \draw[black,dashed] (b1) -- (a2);
    \draw[black,dashed] (b2) -- (a1);

    \draw[blue] (a11) -- (b22);
    \draw[red] (a22) -- (b11);
    \draw[red] (a11) -- (a11_1);
    \draw[red] (b11_1) -- (b11_2);
    \draw[blue] (a22) -- (b11_2);
    
    \draw[blue] (b11) -- (a22_1);
    \draw[red] (a22_1) -- (a22_2);
    \draw[red] (b22) -- (b22_1);

    \draw[blue] (b11_1) -- (a11_1);
    \draw[blue] (b22_1) -- (a22_2);
    \draw[black,dashed] (b11) -- (a11);
    \draw[black,dashed] (b22) -- (a22);
    
  \node[anchor=north] at (a1) {$a_1$};
  \node[anchor=south] at (b1) {$b_1$};
  \node[anchor=south] at (a2) {$a_2$};
  \node[anchor=north] at (b2) {$b_2$};

  \node[anchor=north] at (a11) {$a_1$};
  \node[anchor=south] at (b11) {$b_1$};
  \node[anchor=south] at (a22) {$a_2$};
  \node[anchor=north] at (b22) {$b_2$};
  \node[] at (C_1) {$C_1$};
  \node[] at (C_2) {$C_2$};
  \node[] at (C_11) {$C_1$};
  \node[] at (C_22) {$C_2$};

\old{\node[anchor=east] at (a1_1) {$a_1'$};
  \node[anchor=west] at (b1_1) {$b_1'$};
  \node[anchor=west] at (b1_2) {$b_1''$};

  \node[anchor=west] at (a2_1) {$a_2'$};
  \node[anchor=east] at (a2_2) {$a_2''$};

  \node[anchor=east] at (b2_1) {$b_2'$};
}

  \foreach \p in {a1,b1,a2,b2,a1_1, b1_1, b1_2, a2_1, a2_2, b2_1, a11, b11, a22, b22, a11_1, b11_1, b11_2, a22_1, a22_2, b22_1}
    \filldraw (\p) circle (2pt);
    \end{tikzpicture}
    \caption{Construction of $M'_1$ and $M'_2$, when the crossing edges are part of one cycle}
    \label{fig:crossing}
\end{figure*}

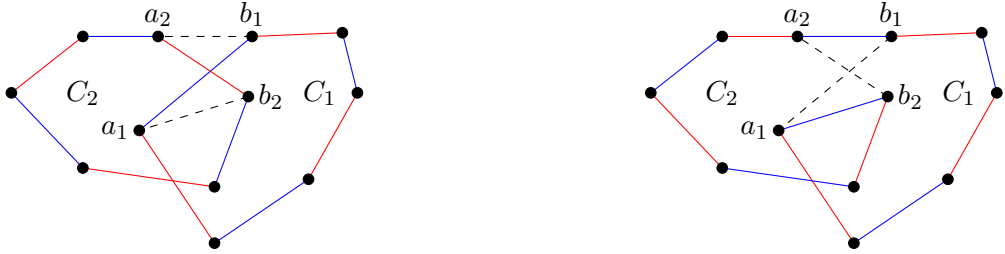
\begin{figure*}[h]
    \centering\begin{tikzpicture}
  \coordinate (a1) at (\x,\y);
  \coordinate (b1) at (\x+1.5,\y+1.25);
  \coordinate (c1) at (\x+1,\y-1.5);
  \coordinate (d1) at (\x+2.25,\y-0.65);
  \coordinate (e1) at (\x+2.9,\y+0.5);
  \coordinate (f1) at (\x+2.7,\y+1.3);  

  \coordinate (a2) at (\x+0.25, \y+1.25);
  \coordinate (b2) at (\x+1.45,\y+0.45);
  \coordinate (c2) at (\x+1,\y-0.75);
  \coordinate (d2) at (\x-0.75,\y-0.5);
  \coordinate (e2) at (\x-1.7,\y+0.5);
  \coordinate (f2) at (\x-0.75,\y+1.25);
    
  \coordinate (C_2) at (\x-0.75,\y+0.5);
  \coordinate (C_1) at (\x+2.4,\y+0.50);

    \draw[blue] (a1) -- (b1);
    \draw[red] (b1) -- (f1);
    \draw[blue] (f1) -- (e1);
    \draw[red] (d1) -- (e1);
    \draw[blue] (c1) -- (d1);
    \draw[red] (a1) -- (c1);
    
    \draw[red] (a2) -- (b2);
    \draw[blue] (b2) -- (c2);
    \draw[red] (c2) -- (d2);
    \draw[blue] (d2) -- (e2);
    \draw[red] (e2) -- (f2);
    \draw[blue] (f2) -- (a2);

    \draw[black,dashed] (b1) -- (a2);
    \draw[black,dashed] (b2) -- (a1);


  \coordinate (a11) at (\x+8.5,\y);
  \coordinate (b11) at (\x+10,\y+1.25);
  \coordinate (c11) at (\x+9.5,\y-1.5);
  \coordinate (d11) at (\x+10.75,\y-0.65);
  \coordinate (e11) at (\x+11.4,\y+0.5);
  \coordinate (f11) at (\x+11.2,\y+1.3);  

  \coordinate (a22) at (\x+8.75, \y+1.25);
  \coordinate (b22) at (\x+9.95,\y+0.45);
  \coordinate (c22) at (\x+9.5,\y-0.75);
  \coordinate (d22) at (\x+7.75,\y-0.5);
  \coordinate (e22) at (\x+6.8,\y+0.5);
  \coordinate (f22) at (\x+7.75,\y+1.25);
    
  \coordinate (C_22) at (\x+7.75,\y+0.5);
  \coordinate (C_11) at (\x+10.9,\y+0.50);

    \draw[blue] (a11) -- (b22);
    \draw[red] (b11) -- (f11);
    \draw[blue] (f11) -- (e11);
    \draw[red] (d11) -- (e11);
    \draw[blue] (c11) -- (d11);
    \draw[red] (a11) -- (c11);
    
    \draw[blue] (a22) -- (b11);
    \draw[red] (b22) -- (c22);
    \draw[blue] (c22) -- (d22);
    \draw[red] (d22) -- (e22);
    \draw[blue] (e22) -- (f22);
    \draw[red] (f22) -- (a22);
    
    \draw[black,dashed] (b22) -- (a22);
    \draw[black,dashed] (b11) -- (a11);

  \node[anchor=east] at (a1) {$a_1$};
  \node[anchor=south] at (b1) {$b_1$};
  \node[anchor=south] at (a2) {$a_2$};
  \node[anchor=west] at (b2) {$b_2$};

  \node[anchor=east] at (a11) {$a_1$};
  \node[anchor=south] at (b11) {$b_1$};
  \node[anchor=south] at (a22) {$a_2$};
  \node[anchor=west] at (b22) {$b_2$};
  \node[] at (C_1) {$C_1$};
  \node[] at (C_2) {$C_2$};
  \node[] at (C_11) {$C_1$};
  \node[] at (C_22) {$C_2$};

  \foreach \p in {a1,b1,a2,b2,c1,c2,d1,d2,e1,e2,f1,f2,a11,b11,a22,b22,c11,c22,d11,d22,e11,e22,f11,f22}
    \filldraw (\p) circle (2pt);
    \end{tikzpicture}
    \caption{Construction of $M'_1$ and $M'_2$, when the crossing edges are part of two different cycles}
    \label{fig:crossing2}
\end{figure*}